\documentclass[11pt,reqno]{amsart}

\usepackage{amscd,amssymb,amsmath,amsthm}
\usepackage{tikz}
\usepackage{graphicx}
\usepackage{color}
\usepackage{cite}
\newtheorem{thm}{Theorem}

\newtheorem{rk}{Remark}

\numberwithin{equation}{section} 

\newcommand{\bea}{\begin{eqnarray}}
\newcommand{\eea}{\end{eqnarray}}

\newcommand{\Z}{\mathbb{Z}}
\newcommand{\Q}{\mathbb{Q}}

\def\Z{\mathbb{Z}}

\begin{document}
\title[Construction of a set of p-adic distributions]{Construction of a set of p-adic distributions}

\author{U. A. Rozikov, Z. T. Tugyonov}

\address{U.\ A.\ Rozikov and Z.\ T.\ Tug'yonov\\ Institute of mathematics,
29, Do'rmon Yo'li str., 100125, Tashkent, Uzbekistan.}
\email {rozikovu@yandex.ru \ \ zohidboy@mail.ru }

\begin{abstract} In this paper adapting to $p$-adic case some methods of real valued Gibbs measures on
Cayley trees  we construct several $p$-adic distributions on the set $\mathbb{Z}_p$ of $p$-adic integers.
Moreover, we give conditions under which these $p$-adic distributions become $p$-adic measures 
(i.e. bounded distributions).
 \end{abstract}
\maketitle



\section{Introduction} \label{sec:intro}

 A $p$-adic distribution is an analogue of ordinary distributions that takes values in a ring of $p$-adic numbers \cite{29}.
 Analogically to a measure on a measurable space, a $p$-adic measure is a special case of a $p$-adic distribution, i.e., a p-adic distribution taking values in a normed space is called a $p$-adic measure if the values on compact open subsets are bounded.

The main purpose for constructing $p$-adic measures is to integrate $p$-adic valued functions.
The development parallels the classic treatment, beginning with the
idea of Riemann sums. Note that such sums may not converge (even for continuous functions) if instead of $p$-adic measures one uses
unbounded $p$-adic distributions.
The theory of $p$-adic measures also useful in the context of $p$-adic $L$-functions following
the works of B. Mazur (see \cite{29} and \cite{GG} for details).

One of the first applications of $p$-adic numbers in quantum physics
appeared in the framework of quantum logic \cite{8}.
Note that this model cannot be
described by using conventional real-valued probability measures.
In (real valued probability) measure theory Kolmogorov's extension theorem (see, e.g.,
\cite[Ch.\,II]{Shiryaev}) is very helpful to construct measures.

A non-Archimedean analogue of the Kolmogorov's theorem was proved in
\cite{16}. Such a result allows to construct wide classes of $p$-adic distributions,
stochastic processes and the possibility to develop statistical
mechanics in the context of $p$-adic theory.

We refer the reader to \cite{GRR}, \cite{19}, \cite{26},\cite{23}-\cite{M} where
various $p$-adic models of statistical physics and $p$-adic distributions
are studied.

In the present paper we construct several $p$-adic distributions and measures on the set  $\mathbb{Z}_p$ of $p$-adic integers.
 We use some arguments of construction of real valued Gibbs measures on
Cayley trees. To do this we present each element of the set  $\mathbb{Z}_p$ by a path in the half (semi-infinite) Cayley tree.
  Then using known distribution relation (see equation (\ref{e2}) below) we show that each
  collection of given $p$-adic distributions can be used to construct new distributions. Moreover, we give conditions under which these $p$-adic distributions become $p$-adic measures.

\section{Preliminaries}

\subsection{$p$-adic numbers.} Let $\Q$ be the field of rational numbers. For a fixed prime number $p$, every rational number $x\ne 0$ can be represented
in the form $x = p^r{n\over m}$, where $r, n\in \Z$, $m$ is a positive integer, and $n$ and $m$ are relatively prime with $p$: $(p, n) = 1$, $(p, m) = 1$. The $p$-adic norm of $x$ is given by
$$|x|_p=\left\{\begin{array}{ll}
p^{-r}\ \ \mbox{for} \ \ x\ne 0\\
0\ \ \mbox{for} \ \ x = 0.
\end{array}\right.
$$
This norm is non-Archimedean  and satisfies the so called strong triangle inequality
$$|x+y|_p\leq \max\{|x|_p,|y|_p\}.$$

The completion of $\Q$ with respect to the $p$-adic norm defines the $p$-adic field
 $\Q_p$. Any $p$-adic number $x\ne 0$ can be uniquely represented
in the canonical form
\begin{equation}\label{ek}
x = p^{\gamma(x)}(x_0+x_1p+x_2p^2+\dots),
\end{equation}
where $\gamma(x)\in \Z$ and the integers $x_j$ satisfy: $x_0 > 0$, $0\leq x_j \leq p - 1$. In this case $|x|_p = p^{-\gamma(x)}$.

The elements of the set $\mathbb{Z}_p=\{x\in \Q_p: |x|_p\leq 1\}$ are called $p$-adic integers.

We refer the reader to \cite{29, 41, 48} for the basics of $p$-adic analysis and $p$-adic mathematical physics.

\subsection{$p$-adic distribution and measure}

Let $(X,{\mathcal B})$ be a measurable space, where ${\mathcal B}$ is an algebra of subsets of $X$.
A function $\mu: {\mathcal B}\to \Q_p$
is said to be a $p$-adic {\it distribution} if for any $A_1, . . . ,A_n\in {\mathcal B}$
such that $A_i\cap A_j = \emptyset$, $i\ne j$, the following holds:
$$\mu(\bigcup^n_{j=1}A_j)=\sum^n_{j=1}\mu(A_j).$$

A $p$-adic distribution $\mu$ is called {\it measure} if it is bounded, i.e.
$$ |\mu(A)|_p\leq +\infty, \ \ \forall A\in \mathcal B. $$
A measure is called a probability measure if $\mu(X) =
1$, see, e.g. \cite{22}, \cite{40}.

In the metric space $\Q_p$ (with metric $\rho(x,y)=|x-y|_p$) basis of open sets can be taken as the following
system of sets:
$$a+(p^N)=\left\{x\in \Q_p: |x-a|_p\leq {1\over p^N}\right\},\ \ a\in \Q_p, \ \ N\in \mathbb{Z}.$$
Such a set is called an interval.
The following theorem is known:

\begin{thm}\label{t3.1.} (see \cite{29}) Every map $\mu$ of the set of
intervals contained in $X\subset \mathbb{Q}_p$, for
which
\begin{equation}\label{e2}
\mu(a+(p^n))=\sum\limits_{b=0}^{p-1}\mu(a+bp^n+(p^{n+1}))
\end{equation} for any $a+(p^n)\subset X$, can be uniquely extended to a
$p$-adic distribution on $X$. \end{thm}

Equation (\ref{e2}) is called a distribution relation.

\begin{rk} If $\mu_1,\mu_2,\dots,\mu_m$ are given distributions with values in $\mathbb{Q}_p$ and
$\lambda_1, \lambda_2, \dots, \lambda_m\in \mathbb{Q}_p$ are arbitrary constants. Then
\begin{equation}\label{ee}
\mu=\sum_{i=1}^m\lambda_i\mu_i
\end{equation}
is a distribution. Indeed, define
$$\mu(a+(p^n))=\sum_{i=1}^m\lambda_i\mu_i(a+(p^n))$$
then $\mu$ satisfies (\ref{e2}) since each $\mu_i$ satisfies (\ref{e2}).
In the theory of real valued measures if one considers a convex set
of such measures then to describe all points of this set it will be sufficient to know all
extreme points of the set. In such a situations measures of the form (\ref{ee}) are not extreme and can be considered
as linear combination of extreme measures.
In the theory of (real) Gibbs measures for a given Hamiltonian (see \cite{17} for details) the set of all such measures
is a convex compact set. To construct new extreme points of this set in \cite{y2} the authors used some known extreme
points. In this paper we will use idea of \cite{y2} to construct (see Theorem \ref{t2} below) new $p$-adic distributions
from the known ones. Proof of Theorem \ref{t3} also reminds the construction of a limiting Gibbs measure with a given boundary condition
on a Cayley tree (see \cite{R} for details of the theory of Gibbs measures on trees).
\end{rk}

\subsection{Cayley tree.} The Cayley tree (Bethe lattice \cite{y1}) $\Gamma^k$
of order $ k\geq 1 $ is an infinite tree, i.e., a graph without
cycles, such that exactly $k+1$ edges originate from each vertex.
Let $\Gamma^k=(V, L)$ where $V$ is the set of vertices and  $L$ the set of edges.
Two vertices $x$ and $y$ are called {\it nearest neighbors} if there exists an
edge $l \in L$ connecting them.
We will use the notation $l=\langle x,y\rangle$.
A collection of nearest neighbor pairs $\langle x,x_1\rangle, \langle x_1,x_2\rangle,...,\langle x_{d-1},y\rangle$ is called a {\it
path} from $x$ to $y$. The distance $d(x,y)$ on the Cayley tree is the number of edges of the shortest path from $x$ to $y$.

If an arbitrary edge $\langle x^0, x^1\rangle=l\in L$ is deleted from the Cayley tree $\Gamma^k$, it splits into two components-- two semi-infinite trees $\Gamma^k_0$ and $\Gamma^k_1$ (each of which is called a half tree). In this paper we consider
half tree  $\Gamma^k_0=(V^0, L^0)$ (see Fig. \ref{f1}).

\begin{figure}
\includegraphics[width=12.5cm]{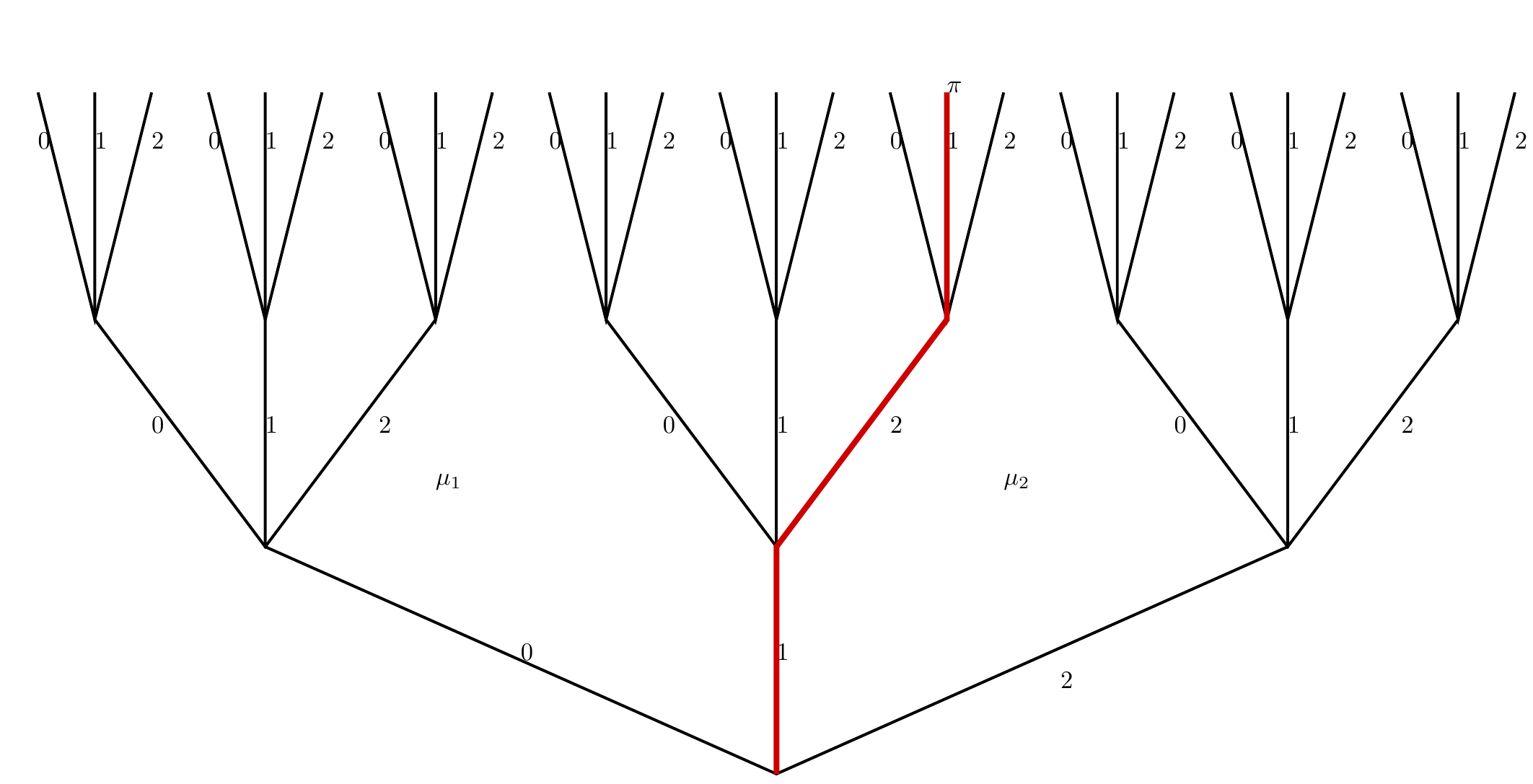}
\caption{\footnotesize \noindent
Half Cayley tree of order three. The bold path $\pi$ has representation by sequence $121\dots$. On the left (resp. right) side of this path the distribution $\mu$ coincides with $\mu_1$ (resp. $\mu_2$).}\label{f1}
\end{figure}

For a fixed $x^0\in V^0$, called the root, we set
\begin{equation*}
W_n=\{x\in V^0| d(x,x^0)=n\}, \qquad V^0_n=\bigcup_{m=0}^n W_m.
\end{equation*}

On the tree $\Gamma^k$ one can introduce a partial ordering, by saying that $y>x$ if there exists
a path $x=x_0,x_1,...,x_n=y$ from $x$ to $y$ that ``goes upwards'', i.e.,
such that $d(x_m,x^0)=d(x_{m-1},x^0)+1, m=1,\, \dots,n.$ The set of vertices $V_x=\{y\in V^0| y\geq x\}$
and the edges connecting them form the semi-infinite tree $\Gamma^k_x$ ``growing'' from the vertex $x\in V^0$.

Let $p$ be a prime number. In this paper we consider the half Cayley tree of order $p$.
Take  an arbitrary (finite or infinite) path $x^0=x_0< x_1 <x_2 < \dots$ on $\Gamma_0^p$
starting from the point $x^0$. We can represent the
path by a sequence $i_0i_1i_2\dots$, where $i_n = 0,1,\dots, p-1$.
Namely, we label by $l_0(x), l_1(x),\dots, l_{p-1}(x)$ the edges
going upward from the vertex $x\in V^0$. Then the path $x^0=x_0<
x_1 < x_2 < \dots$ can be assigned a sequence $i_0i_1\dots$ such
that $\langle x_{n-1},x_n\rangle = l_{i_{n-1}}(x_{n-1})$, $n
=1,2,\dots$; the sequence $i_0i_1i_2\dots$ unambiguously
determines the path $x^0=x_0< x_1 < x_2 < \dots$.

A finite path $x^0=x_0< x_1 < x_2 < \dots<x_n=x$ of length $n$
determines a point $x\in W_n$. If the path $x^0=x_0< x_1 < x_2 <
\dots < x_n=x$  is represented by the sequence $i_0i_1\dots i_n$,
we assume that the vertex $x$ is also represented by this
sequence (see Fig. \ref{f1}).

Let $x,y\in W_n$, and let $x$ be represented by the sequence
$i_0i_1\dots i_n$ and $y$ by $j_0j_1\dots j_n$, where $i_0 = j_0$,
$i_1 = j_1$,\dots, $i_m =j_m$, but $i_{m+1} < j_{m+1}$ for some
$m$. In this case, we write $x\prec y$. Similarly, for pathes $\pi_1$
represented by a sequence $i_0i_1\dots i_n...$ and $\pi_2$ represented by $j_0j_1\dots j_n...$,
we write $\pi_1\prec \pi_2$ if  $i_0 = j_0$,
$i_1 = j_1$,\dots, $i_m =j_m$, but $i_{m+1} < j_{m+1}$ for some
$m$.

Let $\pi=\{x^0=x_0<x_1<\dots\}$ be an infinite path, represented
by the sequence $i_0i_1i_2\dots $. We assign the $p$-adic number
$$t=t(\pi)=\sum_{n=0}^{\infty}{i_n p^n}, \ \  t\in \mathbb{Z}_p $$
to the path $\pi$.  This assignment is a 1-1 correspondence.
Moreover, the partial order $\prec$ on the set of paths can
be used to obtain partial order on $\mathbb{Z}_p$: namely, for $t_1,t_2\in \mathbb{Z}_p$ we write
$t_1\prec t_2$ iff they correspond to $\pi_1$,$\pi_2$ with $\pi_1\prec \pi_2$.

\section{Construction of distributions}

Several examples of $p$-adic distributions are known: Haar's, Mazur's, Bernoulli's distributions etc.\cite{29}.
Recently some periodic distributions were constructed in \cite{T}-\cite{T2}.

Here for any two $p$-adic distribution given on $\mathbb{Z}_p$, say $\mu_1$ and $\mu_2$, $\mu_1\ne \mu_2$ we construct a new distribution
$\mu=\mu(\mu_1,\mu_2)$ defined on $\mathbb{Z}_p$.

Fix a path $\pi$ with representation $\pi=i_0i_1\dots$. Take $a\in \mathbb{Z}_p$.
If there exists $m\geq 0$ such that
$a$ has representation $a=i_0+i_1p+\dots+i_mp^m+a_{m+1}p^{m+1}+\dots$ with $a_{m+1}\ne i_{m+1}$, then we denote this $m$ by
$m(a,\pi)$.
We say that two pathes starting from $x^0$ has non-empty intersection if they have at least one common edge.

\begin{thm}\label{t2} Let $\pi=\{x^0=x_0<x_1<\dots\}$ be an infinite path, represented
by the sequence $i_0i_1i_2\dots $ and $\mu_1$, $\mu_2$ are given distributions on $\mathbb{Z}_p$ such that $\mu_1\ne\mu_2$,
\begin{equation}\label{sh}
\mu_1(a+(p^n))=\mu_2(a+(p^n)), \forall n\leq m(a,\pi)+1,
\end{equation}
then there is a distribution $\mu$ on $\mathbb{Z}_p$, which on the intervals is defined as follows
\begin{equation}\label{mu}
\mu(a+(p^n))=\left\{\begin{array}{lllll}
\mu_1(a+(p^n)), \ \ \mbox{if} \ \ \pi_a\prec \pi, \pi_a\cap \pi=\emptyset \ \ \mbox{or} \ \ \pi_a\cap \pi\ne\emptyset, \ \ n\leq m\\
\ \ \ \ \ \ \ \ \ \ \ \ \ \ \ \ \ \ \ \ \ \mbox{or} \ \ \pi_a\cap \pi\ne\emptyset, \ \ n\geq m+1, \, a_{m+1}<i_{m+1}\\
\mu_2(a+(p^n)), \ \ \mbox{if} \ \ \pi\prec \pi_a, \pi_a\cap \pi=\emptyset \ \ \mbox{or}\\
   \ \ \ \ \ \ \ \ \ \ \ \ \ \ \ \ \ \ \ \ \ \ \pi_a\cap \pi\ne\emptyset, \ \ n\geq m+1, \, a_{m+1}>i_{m+1}\\
\end{array}
\right.
\end{equation}
where $m=m(a,\pi)$ and $\pi_a$ is the path representing $a\in \mathbb{Z}_p$ (see Fig. \ref{f1}).
\end{thm}
\begin{proof}
Consider the following (all possible) cases:

{\it Case I: $\pi_a\cap \pi=\emptyset$}. In this case the path $\pi_a$ always remains on the ``left side" (i.e. $a_0<i_0$)
or always on the ``right side" (i.e. $a_0>i_0$)
from the path $\pi$. For any $n\geq 1$ and any $b=0,1,\dots,p-1$ the number
 $bp^n$ (present in formula (\ref{e2})) makes contributions only on
 levels $n$, $n+1$, ... of the branch $V_{a_0}$ growing from
 the vertex of the edge numbered by $a_0$, (where $a_0\ne i_0$) of the tree $\Gamma^p_0$.
 We note that $V_{a_0}$ does not intersect with $\pi$. Consequently
 we have $\pi_{a+bp^n}\cap \pi=\emptyset$. By definition of $\mu$ given in (\ref{mu})
 in case $\pi_{a}\cap \pi=\emptyset$ and $\pi_{a+bp^n}\cap \pi=\emptyset$ for any $b=0,1,...,p-1$
 one easily checks that $\mu$ always coincides with $\mu_1$ (or $\mu_2$), hence it satisfies the distribution relation (\ref{e2}).

 In case $n=0$, the exists $b_0\in \{0,1,\dots,p-1\}$ such that $\pi_{a+b_0}\cap \pi\ne\emptyset$,
 i.e. $a+b_0=i_0+a_1p+...$, where $i_0$ is the first number in the presentation of $\pi$. In this case assuming $\pi_a\prec \pi$
 the equality (\ref{e2}) can be written as
   \begin{equation}\label{el}
\mu_1(a+(p^0))=\sum\limits_{b=0}^{b_0}\mu_1(a+b+(p^{1}))+\sum\limits_{b=b_0+1}^{p-1}\mu_2(a+b+(p^{1})).
\end{equation}
Since the equality (\ref{e2}) is true for $\mu_1$, the equality (\ref{el}) is satisfied iff
\begin{equation}\label{z}
\sum\limits_{b=b_0+1}^{p-1}\mu_2(a+b+(p^{1}))=\sum\limits_{b=b_0+1}^{p-1}\mu_1(a+b+(p^{1})).
\end{equation}
The last equality is true by condition (\ref{sh}).

Let now $\pi\prec \pi_a$ then
 the equality (\ref{e2}) can be written as
   \begin{equation}\label{er}
\mu_2(a+(p^0))=\sum\limits_{b=0}^{b_0}\mu_2(a+b+(p^{1}))+\sum\limits_{b=b_0+1}^{p-1}\mu_1(a+b+(p^{1})).
\end{equation}
Similarly as the case $\pi_a\prec \pi$ we obtain condition (\ref{z}).

{\it Case II: $\pi_a\cap \pi\ne\emptyset$}. In this case there exists $m=m(a,\pi)\geq 0$ such that
$a$ has representation $a=i_0+i_1p+\dots+i_mp^m+a_{m+1}p^{m+1}+\dots$ with $a_{m+1}\ne i_{m+1}$.

{\sl Subcase} $n\geq m+2$:
Since the path $a_{m+1}a_{m+2}\dots$ does not intersect with $\pi$,
the equation (\ref{e2}) for $n\geq m+2$ is
$$\mu_1(i_0+i_1p+\dots+i_mp^m+a_{m+1}p^{m+1}+\dots+(p^n))=
$$ $$\sum_{b=0}^{p-1}\mu_1(i_0+i_1p+\dots+i_mp^m+a_{m+1}p^{m+1}+\dots+bp^n+(p^{n+1})), \ \ \mbox{if} \ \ a_{m+1}<i_{m+1};$$
and
$$\mu_2(i_0+i_1p+\dots+i_mp^m+a_{m+1}p^{m+1}+\dots+(p^n))=$$
$$\sum_{b=0}^{p-1}\mu_2(i_0+i_1p+\dots+i_mp^m+a_{m+1}p^{m+1}+\dots+bp^n+(p^{n+1})), \ \ \mbox{if} \ \ a_{m+1}>i_{m+1}.$$
Since both $\mu_1$ and $\mu_2$ are distributions the last two equalities are true.

{\sl Subcase} $n\leq m+1$: In this case by our condition (\ref{sh}) we have $\mu_1=\mu_2$, consequently, the distribution relation (\ref{e2}) is satisfied.
\end{proof}

The following theorem says that one can construct a new distribution using a collection of given distributions.

\begin{thm}\label{t3} Let $k$ be a natural number.
Assume for each $t\in W_{k}$ a distribution $\nu_t$ on $\mathbb{Z}_p$ is given and
 there exists an interval $c+(p^{n_0})$, $n_0\geq k$ such that $\nu_t(c+(p^{n_0}))\ne \nu_s(c+(p^{n_0}))$ for at least one pair $t\ne s$, $t,s\in W_k$.
Then there exists a distribution $\mu^{(k)}=\mu^{(k)}[\nu_t, t\in W_k]$
on $\mathbb{Z}_p$ which is different from $\nu_t$ for each $t\in W_k$.
\end{thm}
\begin{proof}
 The distribution $\mu^{(k)}$ can be
constructed as follows (see Fig. \ref{f2}).
\begin{figure}
\includegraphics[width=12.5cm]{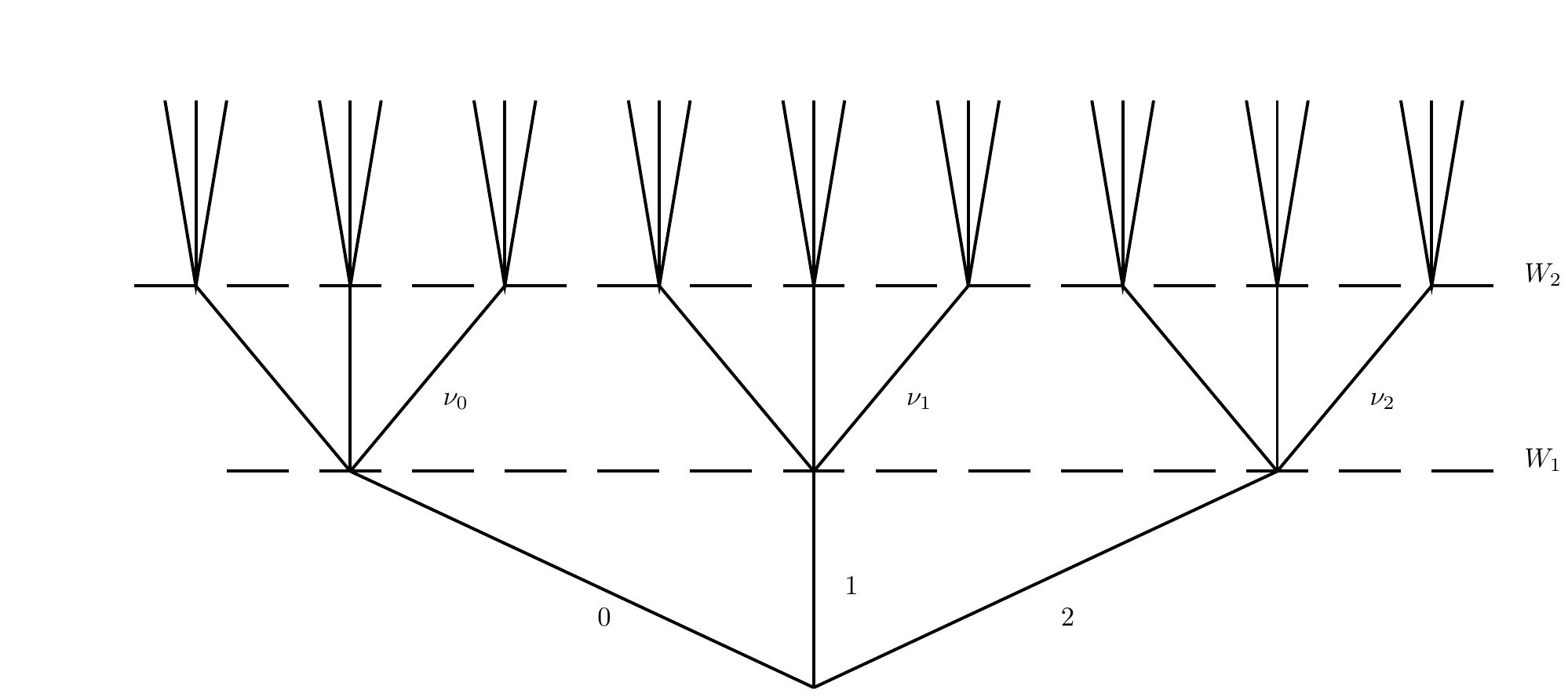}
\caption{\footnotesize \noindent
The half Cayley tree of order three. Case $k=1$, elements of $W_1$ are defined by $0,1,2$ and $\mu^{(1)}$ coincides with distribution $\nu_i$, on the $i$th branch, where  $i=0,1,2$.}\label{f2}
\end{figure}

{\it Case:}  $n\geq k$.  Recall $V_t$ is the set of vertices in the  semi-infinite tree $\Gamma^p_t$ growing from the vertex $t\in V^0$. In the case $n\geq k$ if $\pi_a\cap V_t\ne \emptyset$, for some $t\in W_k$ then $\pi_{a+bp^n}$ also remains in $V_t$ for any $b=0,1,\dots,p-1$.
Therefore we define
\begin{equation}\label{k0}
\mu^{(k)}(a+(p^n))=\nu_t(a+(p^n)), \ \ \mbox{if} \ \ \pi_a\cap V_t\ne\emptyset,
\end{equation}
then $$\mu^{(k)}(a+bp^n+(p^{n+1}))=\nu_t(a+bp^n+(p^{n+1})), \ \ \mbox{for each} \ \  b=0,1,\dots,p-1.$$ Since $\nu_t$ is a given distribution
the equation (\ref{e2}) is satisfied.

{\it Case:} $n\leq k-1$. Without loss of generality we
define $\mu^{(k)}(a+(p^{n}))$ only for $a$ which has the form
$$a=a_0+a_1p+\dots+a_{k-1}p^{k-1},$$
where $a_i\in \{0,1,\dots,p-1\}$.
Note that the sequence $a_0a_1\dots a_{k-1}$ presents a unique point of $W_k$.

For $n=k-1$ define
$$
\mu^{(k)}(a_0+a_1p+\dots+a_{k-1}p^{k-1}+(p^{k-1}))$$
\begin{equation}\label{k1}
=\sum_{b=0}^{p-1}\nu_{a_0a_1\dots a_{k-2}(a_{k-1}+b){\rm mod}\, p} (a_0+a_1p+\dots+[(a_{k-1}+b){\rm mod}\, p] p^{k-1}+(p^{k})).
\end{equation}
Now using formula (\ref{k1})
we define (for $n=k-2$)
$$
\mu^{(k)}(a_0+a_1p+\dots+a_{k-2}p^{k-2}+(p^{k-2}))$$
\begin{equation}\label{k2}
=\sum_{b=0}^{p-1}\mu^{(k)}(a_0+a_1p+\dots+[(a_{k-2}+b){\rm mod}\, p] p^{k-2}+(p^{k-1}))
\end{equation}
Iterating this procedure we recurrently define
$$
\mu^{(k)}(a_0+a_1p+\dots+a_{k-m}p^{k-m}+(p^{k-m}))$$
\begin{equation}\label{km}
=\sum_{b=0}^{p-1}\mu^{(k)}(a_0+a_1p+\dots+[(a_{k-m}+b){\rm mod}\, p] p^{k-m}+(p^{k-m+1})),
\end{equation}
where $m=2, 3, \dots, k$.

From the equalities (\ref{k0})-(\ref{km}) it follows that
$\mu^{(k)}(a+(p^n))$ satisfies (\ref{e2}) for any $a\in \mathbb{Z}_p$
and any $n\geq 0$. Therefore, by Theorem \ref{t2} it follows that $\mu^{(k)}$
can be uniquely extended to a
$p$-adic distribution on $\mathbb{Z}_p$.

Now we show that $\mu^{(k)}$ is different from $\nu_t$ for any $t\in W_k$.
Let us assume the converse,
i.e. there exists $u\in W_k$ such that $\mu^{(k)}=\nu_u$.
This means that for any $a\in \mathbb{Z}_p$ and any $n\in \mathbb{N}$
we have
\begin{equation}\label{uv}
\mu^{(k)}(a+(p^n))=\nu_u(a+(p^n)).
\end{equation}
For any $a\in \mathbb{Z}_p$ there exists $t_a\in W_k$ such that $\pi_a\cap V_{t_a}\ne\emptyset$.
Consequently, from (\ref{k0}) and (\ref{uv}) we get
\begin{equation}\label{uw}
\mu^{(k)}(a+(p^n))=\nu_{t_a}(a+(p^n))=\nu_u(a+(p^n)), \ \ n\geq k.
\end{equation}
When $a$ runs arbitrary $t_a$ takes all values on $W_k$. Therefore from (\ref{uw})
we get that $\nu_{t}(a+(p^n))=\nu_u(a+(p^n))$, $n\geq k$ for any $t\in W_k$.
But this is contradiction to the condition of theorem for $a=c$ and $n=n_0$.
Thus it follows that $\mu^{(k)}\ne \nu_t$ for each $t\in W_k$.
This completes the proof.
\end{proof}
\begin{rk} We note that the distribution $\mu^{(k)}$ mentioned in Theorem \ref{t3}
is useful to construct two distributions $\mu_1$ and $\mu_2$ which satisfy
the condition of Theorem \ref{t2}. Indeed, let $k=1$ and $\pi=i_0i_1\dots$,
consider $\nu_0, \nu_1$ such that $\nu_0(c+(p^{n_0}))\ne \nu_1(c+(p^{n_0}))$, for some $n_0\geq 1$ and $\pi_c\cap \pi=\emptyset$.
Define $\tilde\nu_t$, $t\in W_1$ by
$$\tilde\nu_t=\left\{\begin{array}{ll}
\nu_0, \ \ \mbox{if} \ \ t= i_0\\
\nu_1, \ \ \mbox{if} \ \ t\ne i_0.
\end{array}
\right.$$
Take $\mu_1=\nu_0$ and
$\mu_2=\mu^{(1)}[\tilde\nu_t, t\in W_1]$, it is easy to check that
$\mu_1\ne\mu_2$ and they satisfy the condition of Theorem \ref{t2}.
\end{rk}

In the $p$-adic integral theory one need to $p$-adic measure.
Because integral of
a continuous function calculated with respect to unbounded
$p$-adic distribution may not exist (see \cite{29} for details).

In the following theorem we give a conditions under which
$p$-adic distributions mentioned in Theorems \ref{t2} and \ref{t3}
become $p$-adic measures.

\begin{thm}\label{t4} The following statements hold
\begin{itemize}
\item[1)] Distribution $\mu$,  mentioned in Theorem \ref{t2}, is bounded iff $\mu_1$ and $\mu_2$ are bounded.

\item[2)] Distribution $\mu^{(k)}$, mentioned in Theorem \ref{t3}, is bounded iff $\nu_t(a+(p^n))$ is bounded for each $t\in W_k$ and $n\geq k$.
\end{itemize}
\end{thm}
\begin{proof}
1) This is consequence of the equality (\ref{mu}).

2) Case $n\geq k$: In this case from (\ref{k0}) we get that $\mu^{(k)}$ is bounded iff  $\nu_t$ is bounded for each $t\in W_k$.

Case $n\leq k-1$: Since the right hand side of (\ref{k1}) is sum of bounded terms (corresponding to the case $n=k$) the left hand side also should be bounded. Consequently, iterating of the recursions (\ref{k1})-(\ref{km}) one can see that $\mu^{(k)}$ is bounded iff $\nu_t(a+(p^n))$ is bounded for each $t\in W_k$ and $n\geq k$.
\end{proof}

\section*{ Acknowledgements}

 This work was partially supported by Kazakhstan Ministry of Education and Science, grant 0828/GF4.

\end{document}